\documentclass[11pt,american]{article}
\usepackage{geometry}
\usepackage{amsmath}
\usepackage{amsthm}
\usepackage{amssymb}
\usepackage[usenames, dvipsnames]{xcolor}
\usepackage{float}
 \usepackage{url}
\usepackage{silence}
\usepackage{comment}

\usepackage[colorlinks]{hyperref}
\hypersetup{
linkcolor=BrickRed
,citecolor=Green
,filecolor=Mulberry
,urlcolor=NavyBlue
,menucolor=BrickRed
,runcolor=Mulberry
,linkbordercolor=BrickRed
,citebordercolor=Green
,filebordercolor=Mulberry
,urlbordercolor=NavyBlue
,menubordercolor=BrickRed
,runbordercolor=Mulberry,
pdfauthor={Name}
}
\usepackage{enumitem}
\usepackage[textsize=footnotesize,textwidth=1cm]{todonotes}
\setlist{nosep} 
\setlist[description]{font=\normalfont\space}
\usepackage{tikz} 
\usepackage{extarrows}
\usepackage{complexity}
\theoremstyle{plain}
\newtheorem{lem}{Lemma}[section]
\theoremstyle{definition}

\theoremstyle{definition}

\theoremstyle{plain}

\theoremstyle{remark}

\theoremstyle{plain}
\newtheorem{thm}{Theorem}[section]
\theoremstyle{plain}
\newtheorem{problem}{Problem}[section]
\theoremstyle{plain}
\newtheorem{cor}{Corollary}[section]
\theoremstyle{plain}

\theoremstyle{remark}

\theoremstyle{definition}
\newtheorem{conjecture}{Conjecture}[section]
\theoremstyle{definition}
\newtheorem{remark}{Remark}[section]

\usepackage[nameinlink]{cleveref}

\renewcommand{\ref}{\cref}
\AtBeginDocument{\let\ref\cref}

\crefname{lemma}{Lemma}{Lemmas}
\crefname{observation}{Observation}{Observations}
\crefname{figure}{Figure}{Figures}
\crefname{appendix}{Appendix}{Appendices}
\crefname{lem}{Lemma}{Lemmas}
\crefname{problem}{Problem}{Problems}
\crefname{step}{Step}{Steps}
\crefname{theorem}{Theorem}{Theorems}
\crefname{thm}{Theorem}{Theorems}
\crefname{proposition}{Proposition}{Propositions}
\crefname{conjecture}{Conjecture}{Conjectures}
\crefname{definition}{Definition}{Definitions}
\crefname{fact}{Fact}{Facts}
\crefname{section}{Section}{Sections}
\crefname{corollary}{Corollary}{Corollaries}
\crefname{cor}{Corollary}{Corollaries}
\crefname{alg}{Algorithm}{Algorithms}
\crefname{algorithm}{Algorithm}{Algorithms}
\crefname{ex}{Example}{Examples}
\crefname{exa}{Example}{Examples}
\crefname{example}{Example}{Examples}
\crefname{eq}{Equation}{Equations}
\crefname{remark}{Remark}{Remarks}
\crefname{equation}{Equation}{Equations}
\crefname{subequation}{equation}{equations}

\newcommand{\ie}{, \textit{i}.\textit{e}., }

\newcommand{\splitatcommas}[1]{\begingroup
  \begingroup\lccode`~=`, \lowercase{\endgroup
    \edef~{\mathchar\the\mathcode`, \penalty0 \noexpand\hspace{0pt plus 1em}}}\mathcode`,="8000 #1\endgroup
}
\global\long\def\th#1{#1^{\textrm{th}}}\global\long\def\IZ{\mathbb{Z}}\global\long\def\IN{\mathbb{N}}\global\long\def\IR{\mathbb{R}}\global\long\def\IQ{\mathbb{Q}}\global\long\def\card#1{\left\vert #1\right\vert }\global\long\def\lst#1#2{#1_{1},#1_{2},\dots,#1_{#2}}

\global\long\def\lstc#1#2#3{#1_{1}#3#1_{2}#3\dots#3#1_{#2}}

\global\long\def\pri#1{#1^{\prime}}\global\long\def\ordslp{\operatorname{OrdSLP}}\global\long\def\flr#1{\lfloor#1\rfloor}\global\long\def\paren#1{\left(#1\right)}\global\long\def\degslp{\operatorname{DegSLP}}\global\long\def\mdegslp{\operatorname{mDegSLP}}\global\long\def\acit{\operatorname{ACIT}}

\global\long\def\eqslp{\operatorname{EquSLP}}\global\long\def\bitslp{\operatorname{BitSLP}}\global\long\def\divslp{\operatorname{Div2SLP}}\global\long\def\posslp{\operatorname{PosSLP}}\global\long\def\gtnc{\operatorname{GTNC}}\global\long\def\ord#1{\operatorname{ord}\paren{#1}}
\global\long\def\thrsos{\operatorname{3SoSSLP}}\global\long\def\twosos{\operatorname{2SoSSLP}}\global\long\def\sqslp{\operatorname{SquSLP}}\global\long\def\slp{\operatorname{SLP}}\global\long\def\bin#1{\operatorname{Bin}\paren{#1}}\global\long\def\sss{\operatorname{3SoS}}\global\long\def\ss{\operatorname{2SoS}}

\global\long\def\iu{\mathbb{\imath}}\global\long\def\rev{\operatorname{rev}}\global\long\def\pospolyslp{\operatorname{PosPolySLP}}\global\long\def\sqpolyslp{\operatorname{SqPolySLP}}

\global\long\def\eqdef{:=}\global\long\def\abs#1{\left|#1\right|}\global\long\def\nrm#1{\lVert#1\lVert}

\renewcommand{\BP}{\operatorname{BP}}
\newcommand{\ssp}{\operatorname{SoSRoot}}

\newtheorem{observation}{Observation}[section]
 \newcommand{\markusname}{Markus Bl\"aser}
\newcommand{\markusaffilfull}{Saarland University, Saarland Informatics Campus, Saarbr\"ucken, Germany}

\newcommand{\markusuni}{mblaeser@cs.uni-saarland.de}

\newcommand{\julianname}{Julian D\"orfler}
\newcommand{\julianaffilfull}{Saarland University, Saarland Informatics Campus, Saarbr\"ucken, Germany}

\newcommand{\julianuni}{jdoerfler@cs.uni-saarland.de}

\newcommand{\goravname}{Gorav Jindal}
\newcommand{\goravgmail}{gorav.jindal@gmail.com}

\newcommand{\goravaffil}{Max Planck Institute for Software Systems, Saarland Informatics Campus, Saarbr\"ucken, Germany}

\usepackage{mathptmx}
 
\title{PosSLP and Sum of Squares}
\author{\markusname\thanks{\markusaffilfull.\, Email: \texttt{\markusuni}} \and \julianname \thanks{\julianaffilfull.\, Email: \texttt{\julianuni}} \and \goravname \thanks{\goravaffil.\, Email: \texttt{\goravgmail}}}
\date{}

\begin{document}
\maketitle
\begin{abstract}
The problem PosSLP is the problem of determining whether a given straight-line program (SLP)  computes a positive integer. PosSLP was introduced by Allender et al.\ to study the complexity of numerical analysis (Allender et al., 2009). PosSLP can also be reformulated as the problem of deciding whether the integer computed by a given SLP can be expressed as the sum of squares of four integers, based on the well-known result by Lagrange in 1770, which demonstrated that every natural number can be represented as the sum of four non-negative integer squares.

In this paper, we explore several natural extensions of this problem by investigating whether the positive integer computed by a given SLP can be written as the sum of squares of two or three integers. We delve into the complexity of these variations and demonstrate relations between the complexity of the original PosSLP problem and the complexity of these related problems. Additionally, we introduce a new intriguing problem called Div2SLP and illustrate how Div2SLP is connected to DegSLP and the problem of whether an SLP computes an integer expressible as the sum of three squares.

By comprehending the connections between these problems, our results offer a deeper understanding of decision problems associated with SLPs and open avenues for further exciting research.
\end{abstract}
\section{Introduction}
\subsection{Straight Line Programs and PosSLP}
The problem $\posslp$ was introduced in \cite{Allender06onthe} to study the complexity of numerical analysis and relate the computations over the reals (in the so-called Blum-Shub-Smale model, see \cite{SmaleRealCompu1997}) to classical computational complexity. $\posslp$ asks whether a given integer is positive or not. The problem may seem trivial at first glance but becomes highly non-trivial when the given integer is not explicitly provided but rather represented by an implicit expression which computes it. One way to model the implicit computations of integers and polynomials is through the notion of arithmetic circuits and straight line programs (SLPs). 

An arithmetic circuit takes the form of a directed acyclic graph where input nodes are designated with constants 0, 1, or variables $\lst x m$. Internal nodes are labeled with mathematical operations such as addition ($+$), subtraction ($-$), multiplication ($\times$), or division ($\div$). Such arithmetic circuits are said to be \emph{constant-free}. In the algebraic complexity theory literature, usually, one studies arithmetic circuits where constants are arbitrary scalars from the underlying field. But in this paper, we are only concerned with arithmetic circuits that are constant-free.

On the other hand, a straight-line program is a series of instructions that corresponds to a sequential evaluation of an arithmetic circuit. If this program does not contain any division operations, it is referred to as ``division-free''. Unless explicitly specified otherwise, we will exclusively consider division-free straight-line programs. Consequently, straight-line programs can be viewed as a compact representation of polynomials or integers. In many instances, we will be concerned with division-free straight-line programs that do not incorporate variables, representing an integer. Arithmetic circuits and SLPs are used interchangeably in this paper. Now we define the central object of study in this paper.
\begin{problem}[$\posslp$]
Given a straight-line program representing $N \in \IZ$, decide whether $N > 0$.
\end{problem}
An SLP $P$ computing an integer is a sequence $(b_0,\lst b m)$ of integers such that $b_0 = 1$ and $b_i = b_j \circ_i b_k$ for all $i > 0$, where $j,k < i$ and $\circ_i \in \{+, -, \times\}$. Given such an SLP $P$, $\posslp$ is the problem of determining the sign of the integer computed by $P$\ie the sign of $b_m$. Note that we cannot simply compute $b_m$ from a description of $P$ because the absolute value of $b_m$ can be as large as $2^{2^m}$. Therefore, computing $b_m$ exactly might require exponential time. Hence, this brute force approach of determining the sign of $b_m$ is too computationally inefficient. \cite{Allender06onthe} also show some evidence that $\posslp$ might be a hard problem computationally. They do so by showing the polynomial time Turing equivalence of $\posslp$ to the Boolean part of the problems decidable in polynomial time in the  Blum-Shub-Smale (BSS) model and also to the generic task of numerical computation. We briefly survey this relevance of $\posslp$ to emphasize its importance in numerical analysis. For a more detailed discussion, the interested reader is referred to \cite[Section 1]{Allender06onthe}.

The Blum-Shub-Smale (BSS) computational model deals with computations using real numbers. It is a well-explored area where complexity theory and numerical analysis meet. For a detailed understanding, see \cite{SmaleRealCompu1997}. Here we only dscribe the constant-free BSS model. BSS machines handle inputs from $\IR^{\infty}$, allowing polynomial-time computations over $\IR$ to solve ``decision problems'' $L \subseteq \IR^{\infty}$. The set of problems solvable by polynomial-time BSS machines is denoted by $\P^{0}_{\IR}$, see e.g., \cite{BURGISSER2006147}. To relate the complexity class $\P^{0}_{\IR}$ to classical complexity classes, one considers the \emph{boolean part} of $\P^{0}_{\IR}$, defined as: $\BP(\P_{\IR}^{0}) \eqdef \{L \cap \{0,1\}^{\infty} \mid L \in \P_{\IR}^{0}\}$. To highlight the importance of $\posslp$ as a bridge between the BSS model and classical complexity classes, \cite{Allender06onthe} proved the following \cref{thm:posslpbsseq}.
\begin{thm}[Proposition 1.1 in \cite{Allender06onthe}]
\label{thm:posslpbsseq} We have $\P^{\posslp} = \BP(\P_{\IR}^{0})$.
\end{thm}

Another motivation for the complexity of $\posslp$ comes from its connection to the task of numerical computation. Here we recall this connection from \cite{Allender06onthe}. \cite{Allender06onthe} defined the following problem to formalize the task of numerical computation:

\begin{problem}[Generic Task of Numerical Computation ($\gtnc$) \cite{Allender06onthe}]
Given a straight-line program $P$ with $n$ variables, and given inputs $\lst a n$ for $P$ (as floating-point numbers) and an integer $k$ in unary, compute a floating-point approximation of $P(\lst a n)$ with $k$ significant bits.
\end{problem}

The following result was also demonstrated in \cite{Allender06onthe}.
\begin{thm}[Proposition 1.2 in \cite{Allender06onthe}]
\label{thm:posslpgtnceq} $\gtnc$ is polynomial-time Turing equivalent to $\posslp$.
\end{thm}
\subsection{How Hard is $\posslp$?}
$\gtnc$ can be viewed as the task that formalizes what is computationally efficient when we are allowed to compute with arbitrary precision arithmetic. Conversely, the BSS model can be viewed as formalizing computational efficiency where we have infinite precision arithmetic at no cost. \ref{thm:posslpbsseq} and \ref{thm:posslpgtnceq} show that both these models are equivalent to $\posslp$ under polynomial-time Turing reductions. One can also view these results as an indication that $\posslp$ is computationally intractable. Despite this, no unconditional non-trivial hardness results are known for $\posslp$. Still, a lot of important computational problems reduce to $\posslp$. We briefly survey some of these problems now. By the $n$-bit binary representation of an integer $N$ with the condition $\lvert N \rvert < 2^n$, we mean a binary string with a length of $n+1$. This string consists of a sign bit followed by $n$ bits encoding the absolute value of $N$, with leading zeros added if necessary. A very important problem in complexity theory is the $\eqslp$ problem defined as:
\begin{problem}[$\eqslp$, \cite{Allender06onthe}]
Given a straight-line program representing an integer $N$, decide whether $N = 0$.
\end{problem}
$\eqslp$ is also known to be equivalent to arithmetic circuit identity testing ($\acit$) or polynomial identity testing \cite{Allender06onthe}. It is easy to see that $\eqslp$ reduces to $\posslp$: $N \in \IZ$ is zero if and only if $1 - N^2 > 0$.
Recently, a conditional hardness result was proved for $\posslp$ in \cite{burgisser2023hardness}, formalized below.
\begin{thm}[Theorem 1.2 in \cite{burgisser2023hardness}]
If a constructive variant of the radical conjecture of \cite{saxenaradicalstock2018} is true and $\posslp \in \BPP$ then $\NP \subseteq \BPP$.
\end{thm}
As for upper bounds on $\posslp$, $\posslp$ was shown to be in the counting hierarchy $\CH$ in \cite{Allender06onthe}. This is still the best-known upper bound on the complexity of $\posslp$. Another important problem is the sum of square roots, defined as follows:
\begin{problem}[Sum of Square Roots ($\ssp$)]
Given a list $(\lst a  n)$ of positive integers and a list $(\lst \delta  n) \in \{\pm 1\}^n$ of signs, decide if $\sum_{i=1}^{n} \delta_i \sqrt{a_i}$ is positive.
\end{problem}
$\ssp$ is widely recognized and finds applications in computational geometry and various other domains. The Euclidean traveling salesman problem, whose inclusion in $\NP$ is not known, is easily seen to be in $\NP$ relative to $\ssp$. $\ssp$ is conjectured to be in $\P$ in \cite{Malajovich2001AnEV} but this is far from clear. Still, one can show that $\ssp$ reduces to $\posslp$ \cite{TIWARI1992393,Allender06onthe}. There are several other problems related to straight line program which are intimately related to $\posslp$. For instance, the following problems were also introduced in \cite{Allender06onthe}. These problems would be useful in our discussion later.

\begin{problem}[$\bitslp$]
  Given a straight-line program representing $N$, and given $n, i \in \IN$ in binary, decide whether the $\th{i}$ bit of the $n$-bit binary representation of $N$ is 1.
\end{problem}
It was also shown in \cite{Allender06onthe} that $\posslp$ reduces to $\bitslp$. Although we do not know any unconditional hardness results for $\posslp$, $\bitslp$ was shown to be $\#\P$-hard in \cite{Allender06onthe}. Another important problem related to $\posslp$
is the following $\degslp$ problem, which was shown to be reducible to $\posslp$ in \cite{Allender06onthe}.
\begin{problem}[$\degslp$]
Given a straight-line program representing a polynomial $f \in \IZ[x]$ and a natural number $d$ in binary, decide whether $\deg(f) \leq d$.
\end{problem}
The problem $\degslp$ was posed in \cite{Allender06onthe} for multivariate polynomials, here we have considered its univariate version. But these are seen to be equivalent under polynomial time many one reductions \cite[Proof of Proposition 2.3]{Allender06onthe}, we recall this reduction in \ref{sec:mdegslptodegslp}. We also recall the following new problem from \cite{divtrunc21ccc} related to straight line programs, which is important to results in this paper.

\begin{problem}[$\ordslp$]
Given a straight-line program representing a polynomial $f \in \IZ[x]$ and a natural number $\ell$ in binary, decide whether $\ord{f} \geq \ell$. Here, the order of $f$, denoted as $\ord{f}$, is defined to be the largest $k$ such that $x^k \mid f$.
\end{problem}

\subsection{Our Results}

Lagrange proved in 1770 that every natural number can be represented as a sum of four non-negative integer squares \cite[Theorem 6.26]{NivenIvan199101}. Therefore, $\posslp$ can be reformulated as: Given a straight-line program representing $N \in \IZ$, decide if there exist $a,b,c,d \in \IN$ such that $N = a^2 + b^2 + c^2 + d^2$. In light of this rephrasing of $\posslp$, we study the various sum of squares variants of $\posslp$ in  \ref{sec:slpthreesos} and \ref{sec:slptwosos}. To formally state our results, we define these problems now. For convenience, we say that $n \in \IN$ is $\sss$ if it can be expressed as the sum of three squares (of integers). We study the following problem.

\begin{problem}[$\thrsos$]
Given a straight-line program representing $N \in \IZ$, decide whether $N$ is a $\sss$.
\end{problem}
One might expect that $\thrsos$ is easier than $\posslp$, but we show that $\posslp$ reduces to $\thrsos$ under polynomial-time Turing reductions. More precisely, we prove the following  \ref{thm:3soshard} in \ref{sec:slpthreesos}.
\begin{thm}
\label{thm:3soshard} $\posslp \in \P^{\thrsos}$.
\end{thm}
Similarly, we say that $n \in \IN$ is $\ss$ if it can be expressed as the sum of two squares (of integers). We also study the following problem.
\begin{problem}[$\twosos$]
Given a straight-line program representing $N \in \IZ$, decide whether $N$ is a $\ss$.
\end{problem}
These problems $\thrsos$ and $\twosos$ can also be seen as special cases of the renowned Waring problem. The Waring problem has an intriguing history in number theory. It asks whether for each $k \in \mathbb{N}$ there exists a positive integer $g(k)$ such that any natural number can be written as the sum of at most $g(k)$ many $\th{k}$ powers of natural numbers. Lagrange's four-square theorem can be seen as the equality $g(2) = 4$. Later, Hilbert settled the Waring problem for integers by proving that $g(k)$ is finite for every $k$ \cite{HilbertBeweisFD}. Therefore problems $\twosos$ and $\thrsos$ can be seen as computational variants of the Waring problem. These computational variants of the Waring problems are extensively studied in computer algebra and algebraic complexity theory and Shitov actually proved that computing the Waring rank of multivariate polynomials is $\exists\IR$-hard \cite{shitov2016hard}.
For $\twosos$, we prove the following conditional hardness result in \ref{sec:slptwosos}.
\begin{thm}
\label{thm:2soshard} If the generalized Cram{\'e}r conjecture A 
(\ref{conj:gencramconj}) is true, then $\posslp \in \NP^{\twosos}$.
\end{thm}

We also study whether $\thrsos$ can be reduced to $\posslp$. Unfortunately, we cannot show this reduction unconditionally. Hence we study and rely on the following problem $\divslp$, which might be of independent interest. One can view 
$\divslp$ as the variant of $\ordslp$ for numbers in binary.

\begin{problem}[$\divslp$]
  Given a straight-line program representing $N \in \IZ$, and a natural number $\ell$ in binary, decide if $2^{\ell}$ divides $\abs{N}$ \ie if the $\ell$ least significant bits of $\abs{N}$ are zero.
\end{problem}

We show that if we are allowed oracle access to both $\posslp$ and $\divslp$ oracles then $\thrsos$ can be decided in polynomial time, formalized below in \ref{thm:threesostodivslp}. A proof can be found in \ref{sec:slpthreesos}.

\begin{thm}
\label{thm:threesostodivslp} $\thrsos \in \P^{\{\divslp, \posslp\}}$.
\end{thm}

We also study how $\divslp$ is related to other problems related to straight line programs. To this end, we prove the following \ref{thm:divslphardasdegslp} in  \ref{sec:slpthreesos}.

\begin{thm}
\label{thm:divslphardasdegslp} $\ordslp \equiv_{\P} \degslp \leq_{\P} \divslp$.
\end{thm}

As for the hardness results for $\thrsos$ and $\twosos$, we also show that similar to $\posslp$, $\eqslp$ reduces to both $\thrsos$ and $\twosos$. Analogous to integers, we also study the complexity of deciding the positivity of univariate polynomials computed by a given SLP. In this context, we study the following problem.

\begin{problem}[$\pospolyslp$]
Given a straight-line program representing a univariate polynomial $f \in \IZ[x]$, decide if $f$ is positive\ie $f(x) \geq 0$ for all $x \in \IR$.
\end{problem}

We prove that in contrast to $\posslp$, hardness of $\pospolyslp$ can be proved unconditionally, formalized below in \ref{thm:pospolyslpsonphard}.

\begin{thm}
\label{thm:pospolyslpsonphard} $\pospolyslp$ is $\coNP$-hard under polynomial-time many-one reductions.
\end{thm}

In constrast to numbers, every positive polynomial can be written
as the sum of two squares (but only over the reals,
see Section~\ref{sec:polyassos} for a detailed discussion). So $\pospolyslp$ is equivalent 
to the question whether $f$ is the sum of two squares.
To conclude, we motivate and study the following problem (see \ref{sec:polyassos} for 
more details).

\begin{problem}[$\sqpolyslp$]
Given a straight-line program representing a univariate polynomial $f \in \IZ[x]$, decide if $\exists g \in \IZ[x]$ such that $f = g^2$.
\end{problem}

We show in  \ref{sec:polyassos} that $\sqpolyslp$ is in $\coRP$.

\section{\label{sec:slpthreesos}SLPs as Sums of Three Squares}
This section is concerned with studying the complexity of $\thrsos$ and related problems.
\subsection{Lower Bound for $\protect\thrsos$}
In this section, 
we prove \ref{thm:3soshard}. We use the following characterization
of integers which can be expressed as the sum of three squares.
\begin{thm}  [\cite{Legendre1797,gauss1801disquisitiones,Ankeny1957SumsOT,mordell1958representation}]
\label{thm:threesqtheorem}An integer $n$ is $\sss$ if and only
if it is not of the form $4^{a}(8k+7)$, with $a,k\in\IN.$
\end{thm}

\ref{thm:threesqtheorem} informally implies that $\sss$ integers
are ``dense'' in $\IN$ and hence occur very frequently. A useful application of this intuitive high
density of $\sss$ integers is demonstrated below in \ref{lem:nn23sos}. More formally, Landau showed that the asymptotic density of $\sss$ integers in $\IN$ is 5/6  \cite{landau1909einteilung}. To reduce $\posslp$ to $\thrsos$, we shift the given integer (represented
by a given SLP) by a positive number to convert into $\sss$. To this
end we prove the following \ref{lem:nn23sos}.
\begin{lem}
\label{lem:nn23sos}For every $n\in\IN$, at least one element in
the set $\{n,n+2\}$ is $\sss$.
\end{lem}

\begin{proof}
If $n$ is $\sss$ then we are done. Suppose $n$ is not $\sss$,
by using \ref{thm:threesqtheorem} we know that $n=4^{a}(8k+7)$ for some
$a,k\in\IN.$ If $a=0$ then $n=8k+7$ and hence $n+2=8k+9$ is clearly
not of the form $4^{b}(8c+7)$ for any $b,c\in \IN$. If $a>0$ then $n+2=4^{a}(8k+7)+2$
is not divisible by $4$. Hence for $n+2$ of the form $4^{b}(8c+7)$,
we have to have $4^{a}(8k+7)+2=8c+7$. This is clearly impossible
because the LHS is even whereas RHS is odd.
\end{proof}

\begin{lem}
\label{lem:gen3sos}If $M\in\IZ_{+}$ then  $7M^{4}$ not a $\sss$. 
\end{lem}

\begin{proof}
  Suppose $M=4^{a}(4b+c)$ where $a$ is the largest power of $4$ dividing
$M,$ $c=\frac{M}{4^{a}}\pmod4$ and $b=\flr{\frac{M}{4^{a+1}}}$.
We prove the claim by analyzing the following cases.
\begin{description}
\item [{If}] $c=0$ then $M=4^{a}b$ for some $a,b\in\IN$ and 4 does not
divide $b$. Note that here $a>0$, otherwise $c$ cannot be zero
by its definition. Therefore $7M^{4}=4^{4a}\cdot7b^{4}$. Now we can
apply this Lemma recursively on $b$ (which is smaller than $M$) to infer
that $7b^{4}$ is of the form $4^{\alpha}(8\beta+7)$ for some $\alpha,\beta\in\IN$.
Hence $7M^{4}$ is also of this form and thus not a $\sss$ by using
\ref{thm:threesqtheorem}.
\item [{If}] $c=1$ then $7M^{4}=4^{4a}\cdot7\cdot(256b^{4}+256b^{3}+96b^{2}+16b+1)=4^{\alpha}(8\beta+7)$
for some $\alpha,\beta\in\IN$, hence $7M^{4}$ is not a $\sss$ by
using \ref{thm:threesqtheorem}.
\item [{If}] $c=2$ then $7M^{4}=4^{4a+2}\cdot7\cdot(16b^{4}+32b^{3}+24b^{2}+8b+1)=4^{\alpha}(8\beta+7)$
for some $\alpha,\beta\in\IN$, hence $7M^{4}$ is not a $\sss$ by
using \ref{thm:threesqtheorem}.
\item [{If}] $c=3$ then $7M^{4}=4^{4a}\cdot7\cdot(256b^{4}+768b^{3}+964b^{2}+12496b+81)=4^{\alpha}(8\beta+7)$
for some $\alpha,\beta\in\IN$, hence $7M^{4}$ is not a $\sss$ by
using \ref{thm:threesqtheorem}.
\end{description}
\end{proof}
\ref{lem:gen3sos} implies the following $\eqslp$ hardness of $\thrsos$. 
\begin{lem}
\label{lem:3sosseqslphard}$\eqslp\leq_{\P}\thrsos$.
\end{lem}
\begin{proof}
Given a straight-line program representing an integer $N$, we want
to decide whether $N=0$. Suppose $M=N^{2}$. We have $M\geq0$ and
$M=0$ iff $N=0$. By using \ref{lem:gen3sos}, we know that $7M^{4}$
is a $\sss$ iff $M=0$.
\end{proof}
\begin{remark}
 \ref{lem:3sosseqslphard} illustrates that $\thrsos$ is at least as hard as $\eqslp$ under deterministic polynomial time Turing reductions. This may not appear as a very strong result since $\eqslp$ can be decided in randomized polynomial time anyway. However, unconditionally, even $\posslp$ is known to be only $\eqslp$-hard. Moreover, we rely on \ref{lem:3sosseqslphard} in the proof of  \ref{thm:3soshard} below.
\end{remark}

\begin{proof}[Proof of \ref{thm:3soshard}]
Given a straight-line program representing an integer $N$, we want
to decide whether $N>0$. Using an $\eqslp$ oracle, we first check if $N\in\{0,-1,-2\}$.
By using \ref{lem:3sosseqslphard}, these oracle calls to $\eqslp$ can also be simulated
by oracle calls to the $\thrsos$ oracle.
Hence this task belongs to $\P^{\thrsos}$.
If $N\in\{0,-1,-2\}$, then clearly $N>0$ is false and we answer
``No''. Otherwise we check if $N$ is a $\sss$, if it is then clearly
$N>0$ and we answer ``Yes''. If it is not a $\sss$ then we check
if $N+2$ is a $\sss$. If $N+2$ is a $\sss$ then clearly $N>0$
because $N\not\in\{0,-1,-2\}$. If $N+2$ is not a $\sss$, then 
by \cref{lem:nn23sos} we can
conclude that $N<-2$ and hence we answer ``No''.
\end{proof}

\subsection{Upper Bound for $\protect\thrsos$}
Now we prove the upper bound for $\thrsos$, claimed in \ref{thm:threesostodivslp}.
\begin{proof}[Proof of \ref{thm:threesostodivslp}]
Given an $N\in \IZ$ represented by a given SLP, we want to decide if $N$ is a $\sss$.
By using the $\posslp$ oracle, we first check if $N\geq 0$. If $N < 0$ then
we answer ``No''. Hence we can now assume that $N>0$.
By using \ref{thm:threesqtheorem}, it is easy to see that $N$ is not a $\sss$
iff the binary representation $\bin{N}$ of $N$ looks like below:
\[
  N\text{ is not a }\sss\iff \bin{N}=S1110^{t}\text{ where }t\text{ is even and }S\in\{0,1\}^{*}.
\]
By using the $\divslp$ oracle, we compute the number of trailing
zeroes (call it again $t$) in the binary representation of $N$. This can
be achieved by doing a binary search and repeatedly using the $\divslp$
oracle.
If $t$ is not even then $N$ is a $\sss$. Next we construct an SLP
which computes $2^{t}$\ie the number $10^{t}$ in the binary representation.
Such an SLP can be constructed in time $\poly(\log t)$ and is of
size $O(\log^{2}t)$. This can be seen by looking at the binary representation
of $t$ and then using repeated squaring. We have:
\[
  \bin{N+2^{t}}=\pri S0^{t+3}\iff \bin{N}=S1110^{t}\text{ for some }S,\pri S\in\{0,1\}^{*}.
\]
Hence $N$ is not a $\sss$ iff $N+2^{t}$ has $t+3$ trailing
zeroes, which again can be decided using the $\divslp$ oracle.
\end{proof}

\subsection{Complexity of $\protect\divslp$}

In this section, we show a $\degslp$ lower bound for $\divslp$.
To this end, we first prove the following equivalence of $\degslp$
and $\ordslp$. 

\begin{lem}
\label{lem:reversetheslp}Given a straight-line program $P$ of length
$s$ computing a polynomial $f\in\IZ[x]$ , we can compute in $\poly(s)$
time:
\begin{enumerate}
\item A number $m\in \IN$ such that $\deg(f)\leq m\leq2^{s}$.
\item A straight line program $Q$ of length $O(s)$ such that Q computes
the polynomial $x^{m}f\paren{\frac{1}{x}}$.
\end{enumerate}
\end{lem}

\begin{proof}
We generate the desired straight line program $Q$ in an inductive
manner. Namely, if a gate $g$ in $P$ computes a polynomial $R_{g}$
then the corresponding gate in $Q$ computes a number $m_{g}\geq\deg(R_{g})$
(the gate itself does not compute a number, to be precise, but our reduction algorithm
does) and the polynomial $x^{m_{g}}R_{g}\left(\frac{1}{x}\right)\in\IZ[x]$.
It is clear how to do it for leaf nodes. Suppose $g=g_{1}+g_{2}$
is a $+$ gate in $P$. So we have already computed integers $m_{g_{1}},m_{g_{2}}$
and polynomials $x^{m_{g_{1}}}R_{g_{1}}\left(\frac{1}{x}\right),x^{m_{g_{2}}}R_{g_{2}}\left(\frac{1}{x}\right).$
We consider $m_{g}\eqdef m_{g_{1}}+m_{g_{2}}$. We then have: 
\[
x^{m_{g}}R_{g}\left(\frac{1}{x}\right)=x^{m_{g_{2}}}x^{m_{g_{1}}}R_{g_{1}}\left(\frac{1}{x}\right)+x^{m_{g_{1}}}x^{m_{g_{2}}}R_{g_{2}}\left(\frac{1}{x}\right).
\]
We also construct a  straight-line program of length $O(s)$ that simultaneously
computes $x^{m_{h}}$ for all gates $h$ in $P$. With this, we can
compute $x^{m_{g}}R_{g}\left(\frac{1}{x}\right)$ using 3 additional
gates. This implies the straight-line program for $Q$ can be implemented
using only $O(s)$ gates. Similarly for a $\times$ gate $g=g_{1}\times g_{2}$,
we can simply use $x^{m_{g}}R_{g}\left(\frac{1}{x}\right)=x^{m_{g_{1}}+m_{g_{2}}}R_{g_{1}}\left(\frac{1}{x}\right)R_{g_{2}}\left(\frac{1}{x}\right)$
with $m_{g}=m_{g_{1}}+m_{g_{2}}$. By induction it is also clear that
at the top gate $g$, we have $m_{g}\leq2^{s}$. It remains to describe
a straight-line program of length $O(s)$ which computes $x^{m_{h}}$
for all gates $h$ in $P$. Consider the straight-line program $\pri P$
obtained from $P$ by changing every addition gate into a multiplication
gate. If $\pri g$ is a gate in $\pri P$ corresponding to the gate
$g$ in $P$, then one can show via induction that $R_{\pri g}\left(x\right)$
= $x^{m_{g}}$. This gives the desired straight-line program. 
\end{proof}
\begin{lem}
\label{lem:degspltordslp}$\degslp\leq_{\P}\ordslp$.
\end{lem}

\begin{proof}
Suppose we are given a straight line program $P$ of length $s$ computing
a polynomial $f\in\IZ[x]$. By using \ref{lem:reversetheslp}, we compute:
\begin{enumerate}
\item A number $m\in\IN$ such that $\deg(f)\leq m\leq2^{s}$.
\item A straight line program $Q$ of length $O(s^{2})$ such that $Q$
computes the polynomial $x^{m}f\left(\frac{1}{x}\right)\in\IZ[x]$.
\end{enumerate}
Now it clear that: 
\[
\deg(f)\leq d\Longleftrightarrow\ord{x^{m}f\left(\frac{1}{x}\right)}\geq(m-d).
\]
Hence the claim follows.
\end{proof}
Proof of the following \ref{lem:ordspltdegslp} is almost same to that of \ref{lem:degspltordslp}, hence we omit it.
\begin{lem}
\label{lem:ordspltdegslp}$\ordslp\leq_{\P}\degslp$.
\end{lem}

\begin{thm}
$\ordslp\equiv_{\P}\degslp$.
\end{thm}

\begin{proof}
Follows immediately from \ref{lem:degspltordslp} and \ref{lem:ordspltdegslp}.
\end{proof}
\begin{thm}
$\ordslp\equiv_{\P}\degslp\leq_{\P}\divslp$.
\end{thm}

\begin{proof}
  We only need to show that $\ordslp\leq_{\P}\divslp$. Suppose we
are given a straight line program $P$ of length $s$ computing a
polynomial $f\in\IZ[x]$ and $\ell\in\IN$ in binary, we want to decide
if $\ord f\geq\ell$. We know that $\nrm f_{\infty}\leq2^{2^{s}}$,
where $\nrm f_{\infty}$ is the maximum absolute value of coefficients
of $f(x)$. We now construct an SLP which computes $f(B)$ where $B$
is a suitably chosen large integer which we will specify in a moment. 
If $\ord f\geq\ell$ then clearly
$B^{\ell}$ divides $f(B)$. Now consider the case when $\ord f<\ell$
. So we have $f=x^{m}(f_{0}+xg)$ for some $f_0\in\IZ,g\in\IZ[x]$ and $m<\ell$.
In this case we have:
\[
f(B)=B^{m}(f_{0}+Bg(B)).
\]
If  $B$ is chosen large enough then $B$ does not divide $f_{0}+Bg(B)$
and hence $B^{\ell}$ does not divide $f(B)$. It can be verified
that choosing $B=2^{2^{3s}}$ suffices for this argument. It is also not hard to see that 
an SLP for $f(B)$ can be constructed in polynomial time. Hence we conclude:
\[
\ord f\geq\ell\iff2^{\ell2^{3s}}\text{ divides }f(2^{2^{3s}}).
\]
This completes the reduction.
\end{proof}
\begin{problem}
What is the exact complexity of $\divslp$? 
\end{problem}

Now we show that  $\divslp$ is in $\CH$, this claim follows by employing ideas
from \cite{Allender06onthe}.
\begin{lem}
\label{lem:divslpinch} $\divslp$ is in $\CH$.
\end{lem}
\begin{proof}
Given a straight-line program representing $N\in\IZ$, and a natural
number $\ell$ in binary, we want to decide if $2^{\ell}$ divides $\abs N$\ie
if the $\ell$ least significant bits of $\abs N$ are zero. We show
that this can be done in $\coNP^{\bitslp}$. The condition $2^{\ell}\nmid\abs N$
is equivalent to the statement that at least one bit in $\ell$ least
significant bits of $\abs N$ is one. Hence there is a witness of
this statement\ie the index $i\leq\ell$ such that $i^{\textrm{th}}$
bit of $\abs N$ is one. By using the $\bitslp$ oracle, we can verify
the existence of such a witness in polynomial time. Therefore $\divslp\in\coNP^{\bitslp}$.
By using \cite[Theorem 4.1]{Allender06onthe}, we get that $\divslp\in\coNP^{\CH}\subseteq\CH$. 
\end{proof}

In \ref{app:altproof}, we give a more general proof that shows that
``SLP versions'' of problems in dlogtime uniform $\TC_0$ are in $\CH$.

\section{\label{sec:slptwosos}SLPs as Sum of Two and Fewer Squares}

This section is primarily concerned with studying the complexity of $\twosos$. To this end,
we first recall the following renowned \ref{thm:twososchar} which characterizes when a natural number
is a sum of two squares.

\begin{thm}[{\cite[Section 18]{dudley2012elementary}}]
\label{thm:twososchar}An integer $n>1$ is not $\ss$ if and only
if the prime-power decomposition of $n$ contains a prime of the form
$4k+3$ with an odd power.
\end{thm}
When the input integer $n$ is given explicitly as a bit string,  \ref{thm:twososchar} illustrates that a factorization oracle suffices to determine whether $n$ is a $\ss$. In fact, we are not aware of any algorithm that bypasses the need for factorization. For $x\in\IZ_{+}$, let $B(x)$ denote the number of $\ss$
integers in $[x]$. Landau's Theorem \cite{landau1908uber}
gives the following asymptotic formula for $B(x)$.
\begin{thm}[\cite{landau1908uber}]
\label{thm:landau2sosdensity}$B(x)=K\frac{x}{\sqrt{\ln x}}+O\paren{\frac{x}{\ln^{3/2}x}}$
as $x\to\infty$, where $K$ is the Landau-Ramanujan constant with
$K\approx0.764$.
\end{thm}
Ideally, we want to use the above \ref{thm:landau2sosdensity} on the density of $\ss$ to show that
$\posslp$  reduces to $\twosos$, as we did for $\thrsos$. There are two issues with this approach:
\begin{enumerate}
  \item The density of $\ss$ integers is not as high as $\sss$ integers, hence to find 
    the next $\ss$ integer after a given $N\in \IN$ might require a larger shift (as compared to the  shift of 2
    for $\sss$). This issue is overcome below by using $\NP$ oracle reductions instead of $\P$ reductions.
  \item A more serious issue is that \ref{thm:landau2sosdensity}  says something about
    the density of $\ss$ integers only asymptotically, as $x\to\infty$. But
    for this idea of finding the next $\ss$ integer after a given integer only works if this density bound
    is true for all intervals of naturals. This issue is side stepped by relying on the \ref{conj:gencramconj} below.
\end{enumerate}
Let $q$ and $r$ be positive integers such that $1\leq r<q$ and
$\gcd(q,r)=1$. We use $G_{q,r}(x)$ to denote maximum gap between
primes in the arithmetic progression $\{qn+r\mid n\in\IN,qn+r\leq x\}$.
We use $\varphi(n)$ to denote the Euler's totient function\ie the number of positive $m\leq n$ with $\gcd(m,n)=1$.
\begin{conjecture}[Generalized Cram{\'e}r conjecture A, \cite{kourbatov2018distribution}]
\label{conj:gencramconj}For any $q>r\geq1$ with $\gcd(q,r)=1$,
we have 
\[
G_{q,r}(p)=O(\varphi(q)\log^{2}p).
\]
\end{conjecture}

\subsection{Lower Bounds for $\protect\twosos$}
\begin{lem}
\label{lem:2sosseqslphard}$\eqslp\leq_{\P}\twosos$.
\end{lem}

\begin{proof}
Given a straight-line program representing an integer $N$, we want
to decide whether $N=0$. Suppose $M=N^{2}$. We have $M\geq0$ and
$M=0$ iff $N=0$. If $M\neq0$ then by employing \ref{thm:twososchar},
$3M^{2}$ cannot be a $\ss$. Hence $3M^{2}$ is a $\ss$ iff $M=0$. 

\end{proof}

\begin{proof}[Proof of \cref{thm:2soshard}]
Given a straight-line program of size $s$ representing an integer
$N$, we want to decide whether $N>0$. 
Choose $M=2^{3s}$. We compute the value of $N \bmod T$, where $T\eqdef 2M+1$. This can be done in
$\poly(s)$ time by simulating the computation of the given SLP (which computes $N$)
modulo $T$. If $\abs{N} \leq M$, then we can even recover the exact value of $N$ by
knowing $N\bmod T$. So by assuming $\abs{N}\leq M$, we guess the exact value of $N$. 
Let us call this guessed value (computed by knowing $N \bmod T$) to be $\pri{N}$.
Then by using $\eqslp$ oracle (which can be simulated by
$\twosos$ oracle using \ref{lem:2sosseqslphard}), we  check if the equality
$N=\pri{N}$ is actually true. If $N=\pri{N}$ then we can easily determine the sign of $N$. Otherwise our assumption
$\abs{N}\leq M$ is false and hence we can assume that $\abs{N}>M$.

Suppose $p\geq \abs{N}$ is the smallest prime of the form $4k+1$. By using
the results in \cite{Breusch1932}, we know that $p\leq2\abs{N}$ for $\abs{N}\geq7$.
By using \ref{conj:gencramconj} with $q=4,r=1$, we get that $p\leq \abs{N}+O(\varphi(4)\log^{2}p)\leq \abs{N}+c\log^{2}\abs{N}$
for some absolute constant $c$. For large enough $\abs{N}$, it implies that $p\leq \abs{N}+\log^{3}\abs{N}\leq \abs{N}+2^{3s}$.
Therefore $p-\abs{N}\leq M$. By using \ref{thm:twososchar}, we know that $p$ is a $\ss$.
Now we guess the witness $S\eqdef p-\abs{N}$, clearly this has a binary description of size
at most $O(s)$. Now we use $\twosos$ oracle to check if $N+S$ is a $\ss$. 
If $N>0$ then clearly such a witness exists.
On the other hand if $N<0$ then we know that $N < -M$ and $N+S$ cannot be a $\ss$.
Therefore if \ref{conj:gencramconj}is true then $\posslp \in \NP^{\twosos}$.
\end{proof}
Similarly to $\twosos$ and $\thrsos$, one can also study the complexity of the following problem
$\sqslp$.
\begin{problem}[$\sqslp$, Problem 7 in \cite{sumofsqrissac2023}]
Given a straight-line program representing $N\in\IZ$, decide whether
$N=a^{2}$ for some $a\in\IZ$.
\end{problem}

$\sqslp$ was shown to be decidable in randomized polynomial time
in \cite[Sec 4.2]{sumofsqrissac2023}, assuming GRH. 
The complexity of $\twosos$ remains an intriguing open problem. If $\twosos$  were to be in $\P$ then this would disprove \ref{conj:gencramconj} or prove that 
$\posslp \in \NP$, neither of which is currently known.

\section{Polynomials as Sum of Squares}\label{sec:polyassos}
\subsection{Positivity of Polynomials}
Analogous to $\posslp$, we also study the positivity problem for
polynomials represented by straight line programs. In particular,
we study the following problem, called $\pospolyslp$.
\begin{problem}[$\pospolyslp$]
Given a straight-line program representing a univariate polynomial
$f\in\IZ[x]$, decide if $f$ is positive\ie $f(x)\geq0$ for all
$x\in\R$.
\end{problem}

It is known that every positive univariate polynomial $f$ can be written as
sum of two squares. The formal statement (\ref{lem:realpossumoftwosqr})  and its folklore proof can be found in
the appendix.

Now we look at the rational variant of the \ref{lem:realpossumoftwosqr}.
Suppose $f\in\IZ[x]\subset\IQ[x]$ is a positive polynomial. We know
that it can be written as sum of squares of two real polynomials.
Can it also be written as sum of squares of rational polynomials?
In this direction, Landau proved that each positive polynomial in
$\IQ[x]$ can be expressed as a sum of at most eight polynomial squares
in $\IQ[x]$ \cite{Landau1906}. Pourchet improved this result and
proved that only five or fewer squares are needed \cite{Pourchet1971}.

We now show that $\pospolyslp$ is $\coNP$-hard, this result follows
from an application of results proved in \cite{PERRUCCI2007471}.
Suppose $W$ is a 3-SAT formula on $n$ literals $\lst xn$ with $W=C_{1}\land C_{2}\land\dots\land C_{\ell}$,
here $C_{i}$ is a clause composed of 3 literals. We choose any $n$
distinct odd primes $\lstc pn<$. So $x_{i}$ is associated with the
prime $p_{i}$. Thereafter, we define $M\eqdef\prod_{i\in[n]}p_{i}$.
The following \ref{thm:pmwsmallcircuit} was proved in \cite{PERRUCCI2007471}.

\begin{thm}[\cite{PERRUCCI2007471}]
\label{thm:pmwsmallcircuit}One can construct a SLP $C$ of size
$\poly(p_{n},\ell)$ which computes a polynomial $P_{M}(W)$ of the form:
\[
P_{M}(W)\eqdef\sum_{i\in[\ell]}(F_{M}(C_{i}))^{2}
\]
such that $P_{M}(W)$ has a real root iff $W$ is satisfiable. Here, $F_M(C_i)$ is a univariate polynomial that depends on $C_i$ (see \cite{PERRUCCI2007471} for more details).
\end{thm}
\begin{thm}[Theorem 1.2 in \cite{basu2009bound}]
\label{thm:minpolyfval}Let $f\in\IZ[x]$ be a univariate polynomial
of degree $d$ taking only positive values on the interval $[0,1]$.
Let $\tau$ be an upper bound on the bit size of the coefficients
of $f$. Let $m$ denote the minimum of $f$ over $[0,1]$. Then 
\[
m>\frac{3^{d/2}}{2^{(2d-1)\tau}(d+1)^{2d-1/2}}.
\]
\end{thm}

The theorem above proves the lower bound for the interval $[0,1]$.
Next, we extend it to the whole real line.

\begin{lem}
\label{lem:minvalpospoly}Let $f\in\IZ[x]$ be a positive univariate
polynomial of degree $d$. Let $\tau$ be an upper bound on the bit
size of the coefficients of $f$. Let $m$ denote the minimum of $f$
over $\IR$. If $m\neq0$ then 
\[
m>\frac{3^{d/2}}{2^{(2d-1)\tau}(d+1)^{2d-1/2}}.
\]
\end{lem}

\begin{proof}
We assume $m\neq0$. Consider the reverse polynomial $f_{\rev}\eqdef x^{d}f\paren{\frac{1}{x}}$.
It is is clear that $f_{\rev}$ is positive on $[0,\infty)$. Moreover,
$f_{\rev}$ has degree $d$ and $\tau$ is an upper bound on the bit
size of its coefficients. By employing \ref{thm:minpolyfval} on $f_{\rev}$,
we infer that 
\[
\min_{a\in[0,1]}f_{\rev}(a)>\frac{3^{d/2}}{2^{(2d-1)\tau}(d+1)^{2d-1/2}}.
\]
\ref{thm:minpolyfval} implies that:
\begin{equation}
\min_{a\in[0,1]}f(a)>\frac{3^{d/2}}{2^{(2d-1)\tau}(d+1)^{2d-1/2}}.\label{eq:minf01}
\end{equation}
Now consider a $\lambda\in[0,1]$, we have:

\begin{equation}
f\paren{\frac{1}{\lambda}}=\frac{f_{\rev}(\lambda)}{\lambda^{d}}\geq f_{\rev}(\lambda)>\frac{3^{d/2}}{2^{(2d-1)\tau}(d+1)^{2d-1/2}}.\label{eq:minf1inf}
\end{equation}
By combining \ref{eq:minf01} and \ref{eq:minf1inf}, we obtain that:
\[
\min_{a\in[0,\infty)}f(a)>\frac{3^{d/2}}{2^{(2d-1)\tau}(d+1)^{2d-1/2}}.
\]
By repeating the above argument on $f(-x)$ instead of $f(x)$, we
obtain:

\[
m=\min_{a\in(-\infty,\infty)}f(a)>\frac{3^{d/2}}{2^{(2d-1)\tau}(d+1)^{2d-1/2}}.
\]
\end{proof}

\begin{proof}[Proof of \cref{thm:pospolyslpsonphard}]
Suppose $W$ is 3-SAT formula on $n$ literals $\lst xn$ with $W=C_{1}\land C_{2}\land\dots\land C_{\ell}$,
here $C_{i}$ is a clause composed of 3 literals. By using \ref{thm:pmwsmallcircuit},
we can construct a SLP of size $\poly(p_{n},\ell)$ which computes a
polynomial $P(W)\in\IZ[x]$ such that $P(W)$ has a real root iff $W$
is satisfiable. 
(Recall that $p_1 < \dots < p_n$ was a sequence of odd primes.)
Since $P(W)$ is a sum of squares, $P(W)$ is positive.
Suppose $m$ is the minimum value of $P(W)$ over $\IR$. We know that
$m\geq0$.

By the prime number theorem we can assume $p_{n}=O(n\log n)$. Moreover,
it is easy to see that $\ell\leq8n^{3}$. Hence the constructed SLP is
of size $s=\poly(n)$. Suppose $\tau$ is an upper bound on the bit
size of the coefficients of $P(W)$. It is easy to see that $\deg(P(W))\leq2^{s}$
and $\tau\leq2^{s}$. If $W$ is not satisfiable then we know that
$m\neq0$ and therefore \ref{lem:minvalpospoly} implies that 
\[
\log(m)>2^{s-1}\log3-(2^{s+1}-1)2^{s}-(2^{s+1}-1/2)\log(2^{s}+1)>-2^{2s+2}.
\]
Hence
\[
m>\frac{1}{2^{2^{2s+2}}}.
\]
Suppose $B=2^{2^{2s+2}}$. Then $B\cdot P(W)-1$ is positive iff $m>0.$
Hence we have:
\[
B\cdot P(W)-1\text{\text{ is positive} iff }W\text{ is unsatisfiable}.
\]
Moreover $B\cdot P(W)-1$ has a SLP of size $O(s)=\poly(n)$ and this
SLP can be constructed in time $\poly(n)$. Since determining the
unsatisfiability of $W$ is $\coNP$-complete, it follows that $\pospolyslp$
is $\coNP$-hard.
\end{proof}

\subsection{Checking if a Polynomial is a Square}\label{subsec:sesqpolyslp}

In light of \cite{Pourchet1971}'s result and \ref{thm:pospolyslpsonphard},
we also study the following related problem $\sqpolyslp$. Another motivation to study this
problem also comes from the quest for studying the complexity of factors of polynomials.
In this context, one wants to prove that if a polynomial can be computed a small arithmetic circuit then
so can be its factors.  In this direction, Kaltofen showed that if a polynomial
$f=g^eh$  can be computed an arithmetic circuit of size $s$ and $g,h$ are coprime, then
$g$ can also be computed by a circuit of size
$\poly(e, \deg(g), s)$ \cite{Kaltofen87factoring}. When $f=g^e$ ,
Kaltofen also showed that $g$ can be computed an arithmetic circuit of size $\poly(\deg(g), s)$ \cite{Kaltofen87factoring}.
This question for finite fields is asked as an open question in \cite{kopparty2014equivalence}.
What if we do not want to find a small circuit for polynomial $g$ in case
$f=g^e$ but only want to determine if $f$ is $\th{e}$ power
of some polynomial. And in this decision problem, we want to avoid the 
dependency on $\deg(g)$ in running time, which can be exponential in $s$.
We study this problem for $e=2$ in $\sqpolyslp$, but our results work for any arbitrary constant $e$.
\begin{problem}[$\sqpolyslp$]
Given a straight-line program representing a univariate polynomial
$f\in\IZ[x]$, decide if $\exists g\in\IZ[x]$ such that $f=g^{2}$.
\end{problem}
One can also study the complexity of determining if the given univariate polynomial can be 
written as a sum of two, three or four squares but in this section,
we only focus on the problem $\sqpolyslp$. The following \ref{thm:whenisfsqr} 
hints to an approach that $\sqpolyslp$ can be reduced to $\sqslp$.
\begin{thm}[Theorem 4 in \cite{murty2008polynomials}]
\label{thm:whenisfsqr}For $f\in\IZ[x]$, $\exists g\in\IZ[x]$ with
$f=g^{2}$ iff $\forall t\in\IZ$, $f(t)$ is a perfect square.
\end{thm}

We shall use an effective variant of \ref{thm:whenisfsqr} which follows
the following effective variant of the Hilbert's irreducibility theorem.
For an integer polynomial $f$, $H(f)$ is the height of $f$\ie
the maximum of the absolute values of the coefficients of $f$.
\begin{thm}[\cite{walkowiak2005,debes2008}]
\label{thm:hilbertirredbound}Suppose $P(T,Y)$ is an irreducible
polynomial in $\IQ[T,Y]$ with $\deg_{Y}(P)\geq2$ and with coefficients
in $\IZ$ assumed to be relatively prime. Suppose $B$ is a positive
integer such that $B\geq2$. We define:
\begin{align*}
m & \eqdef\deg_{T}(P)\\
n & \eqdef\deg_{Y}(P)\\
H & \eqdef\max(H(P),e^{e})\\
S(P,B) & \eqdef\card{\{1\leq t\leq B \mid P(t,Y)\text{ is reducible in }\IQ[Y]\}}
\end{align*}
Then we have:
\[
S(P,B)\leq2^{165}m^{64}2^{296n}\log^{19}(H)B^{\frac{1}{2}}\log^{5}(B).
\]
\end{thm}

\begin{cor}
\label{cor:sqpolyranbound}Suppose $f(x)\in\IZ[x]$ is an integer polynomial
computed by a SLP of size $s$. Define $S(f)\eqdef\card{\{1\leq t\leq2^{200s}\mid f(t)\text{ is a square }\}}$.
If $f$ is not a square then we have:
\[
S(f)<2^{800}s^{5}2^{183s}.
\]
\end{cor}

\begin{proof}
Consider the polynomial $P(T,Y)\eqdef Y^{2}-f(T)$. Since $f(x)$
is not a square, we infer that $P(T,Y)$ is an irreducible polynomial
in $\IQ[T,Y]$. Now we employ \ref{thm:hilbertirredbound} on $P(T,Y)$
with $B=2^{200s}$, we have $m\leq2^{s},n=2$ and $H\leq2^{2^{s}}.$
In this case, we have $S(P,B)=S(f)$. By using \ref{thm:hilbertirredbound}
, we have:

\[
S(f)\leq2^{165}2^{64s}2^{592}2^{19s}2^{100s}(200s)^{5}<2^{800}s^{5}2^{183s}.
\]
\end{proof}
\ref{cor:sqpolyranbound} implies a randomized polynomial time algorithm
for $\sqpolyslp$, as demonstrated below in \ref{thm:sqpolyslpincorp}. 
\begin{thm}\label{thm:sqpolyslpincorp}
$\sqpolyslp$ is in $\coRP$.
\end{thm}

\begin{proof}
Given an integer polynomial $f(x)$ computed by a SLP of size $s$,
we want to decide if $f=g^{2}$ for some $g\in\IZ[x]$. We sample a
positive integer uniformly at random from the set $\{1\leq t\leq2^{200s}\mid t\in\IN\}$.
Using the algorithm in \cite[Sec 4.2]{sumofsqrissac2023}, we test
if $f(t)$ is a square. We output ``Yes'' if $f(t)$ is a square.
If $f=g^{2}$ for some $g\in\IZ[x]$, then we always output ``Yes''.
Suppose $f\neq g^{2}$ for any $g\in\IZ[x]$. By using \ref{cor:sqpolyranbound}, we obtain that:
\[
\Pr[f(t)\text{ is a square}]<\frac{2^{800}s^{5}2^{183s}}{2^{200s}}<\frac{1}{100}\text{ for }s>100.
\]
Hence with probability at least $0.99$ we sample a $t$ such that
$f(t)$ is not a square. The algorithm for $\sqslp$ verifies that $f(t)$ is not a square with probability at least $\frac{1}{3}$ \cite[Sec 4.2]{sumofsqrissac2023}.
Hence we output ``No'' with probability at least $0.33$. This implies $\sqpolyslp \in \coRP$.
\end{proof}

\section{Conclusion and Open Problems}
We studied the connection between $\posslp  $ and problems related to 
the representation of integers as sums of squares, drawing on Lagrange's 
four-square theorem from 1770. We investigated variants of the problem, 
considering whether the positive integer computed by a given SLP can be
represented as the sum of squares of two or three integers. We analyzed 
the complexity of these variations and established relationships between 
them and the original $\posslp$ problem. Additionally, we introduced the 
$\divslp$ problem, which involves determining if a given SLP computes an 
integer divisible by a given power of 2. We showed that $\divslp$ is at least
as hard as $\degslp$. We also showed the relevance of $\divslp$ in connecting the 
$\thrsos$ to $\posslp$. In contrast to $\posslp$, we also showed that the 
polynomial variant of the $\posslp$ problem is unconditionally $\coNP$-hard. 
Overall, this paper contributes to a deeper understanding of decision problems 
associated with SLPs and provides insights into the computational complexity of
problems related to the representation of integers as sums of squares. 
A visual representation illustrating the problems discussed in this paper and their interrelations is available in \ref{fig:vis}. Our results open avenues for further research in this area; 
in particular, we highlight the following research avenues:
\begin{enumerate}
 \item   What is the complexity of $\divslp$? We showed it is $\degslp$ hard. Is it $\NP$-hard too? How does it relate to $\posslp$?
 \item   Can we prove \ref{thm:2soshard} without relying on \ref{conj:gencramconj}?
 \item   One can also study the problems of deciding whether a given SLP computes an integer univariate polynomial, which can be written as the sum of two, three, or four squares. We studied these questions for integers in this paper. But it makes for an interesting research to study these questions for polynomials.
 \item  And finally, can we prove unconditional hardness results for $\posslp$?
\end{enumerate}
\paragraph{Acknowledgments:}
We would like to thank Robert Andrews for providing a simpler proof of \ref{lem:reversetheslp}. We had a proof of it that was a bit longer. Robert Andrews simplified the proof after reviewing our proof in a personal communication.

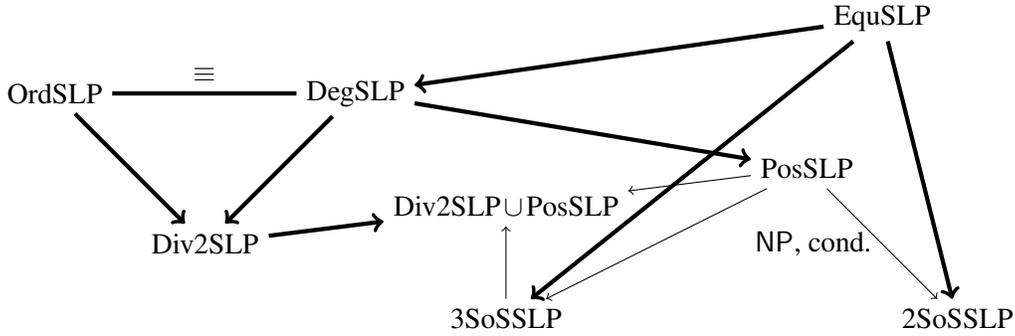
\begin{figure}[H]
\centering
\begin{tikzpicture}\node at (6,0) [] (DivSLP) {$\divslp$};
 \node at (4,2) [] (OrdSLP) {$\ordslp$}; 
 \node at (8,2) [] (DegSLP) {$\degslp$};

 \node at (14,1) [] (PosSLP) {$\posslp$};

 \node at (15,3) [] (EquSLP) {$\eqslp$};
 
 \node at (10,-1) [] (3SoSSLP) {$\thrsos$};
 \node at (16,-1) [] (2SoSSLP) {$\twosos$};

 \node at (10,0.5) [] (Union) {$\divslp\cup\posslp$};

 \draw [ultra thick,-] (OrdSLP) edge node[above] {$\equiv$} (DegSLP); 
 \draw [ultra thick,->] (OrdSLP) -- (DivSLP);
 \draw [ultra thick,->] (DegSLP) -- (DivSLP);
  \draw [ultra thick, ->] (DegSLP) -- (PosSLP);

 \draw [->] (PosSLP) -- (3SoSSLP); 
 \draw [->] (PosSLP) edge node[left] {$\NP$, cond.} (2SoSSLP); 

 \draw[ultra thick, ->] (DivSLP) -- (Union);
 \draw[->] (PosSLP) -- (Union);
 \draw[->] (3SoSSLP) -- (Union);

  \draw[ultra thick, ->] (EquSLP) -- (DegSLP);
  \draw[ultra thick, ->] (EquSLP) -- (3SoSSLP);
  \draw[ultra thick, ->] (EquSLP) -- (2SoSSLP);

\end{tikzpicture}
\caption{A visualization of the relations between the problems studied in this
work. An arrow means that there is a Turing reduction. A thicker arrow indicates a polynomial time many one
reduction.
The reduction from $\posslp$ to $\twosos$ is nondeterministic
and depends on \ref{conj:gencramconj}.  \label{fig:vis}}
\end{figure}

\newcommand{\etalchar}[1]{$^{#1}$}

\appendix
\section{Missing Proofs}
\begin{lem}
\label{lem:realpossumoftwosqr}For every positive polynomial $f\in\IR[x]$,
there exist $g,h\in\IR[x]$ such that $f=g^{2}+h^{2}$
\end{lem}

\begin{proof}
If $f(x)\ge0$ for $x\in\R$ and $\alpha$ is a real root of $f$,
then it must have even multiplicity. We have $(x-\alpha)^{2}=(x-\alpha)^{2}+0^{2}$.
We use $\iu$ to denote $\sqrt{-1}.$ If $\beta=s+\iu t$ and $\overline{\beta}=s-\iu t$
are a complex-conjugate pair of roots of $f$ then, $(x-\beta)(x-\overline{\beta})=(x-s-\iu t)(x-s+\iu t)=(x-s)^{2}+t^{2}$
is a sum of two squares. Now the claim follows by using the identity
$(a^{2}+b^{2})(c^{2}+d^{2})=(ac-bd)^{2}+(ad+bc)^{2}$.
\end{proof}

\section{Alternative Proof of \ref{lem:divslpinch}}
\label{app:altproof}

We prove a general theorem on how to show that problems involving
SLPs are in $\CH$. It is similar to the proof of 
\cite[Lem 5]{DBLP:conf/coco/AllenderKRRV01}.
Let $C$ be a Boolean circuit in $\TC_{0}$. $C$ consists of unbounded
AND, unbounded OR, and unbounded majority gates (MAJ). According to
\cite{DBLP:journals/jcss/Ruzzo81}, when a family $(C_{n})$ is in dlogtime-uniform
$\TC_{0}$, this means that there is a deterministic  Turing machine (DTM)
that decides in time $O(\log n)$  whether
given $(n,f,g)$ the gate $f$ is connected to the gate $g$ and whether
given $(n,f,t)$ the gate $f$ has type $t$. All numbers are given in
binary. For a language $B\subseteq\{0,1\}^*$, let $\slp(B)$ be the
language:
\[
  \slp(B)\eqdef \{P\mid\text{\ensuremath{P} is an SLP computing a number \ensuremath{N} such that \ensuremath{\bin{N}\in B} (as a binary string)}\}
\]
This can be viewed as the ``SLP-version'' of $B$.

\begin{lem}
Let $B$ be in dlogtime-uniform $\TC_{0}$. Then $\slp(B)\in\CH$. 
\end{lem}

\begin{proof}
The proof is by induction on the depth. We prove the more general
statement: Let $M$ be a DTM from the definition of dlogtime and
$(C_{n})$ be the sequence of circuits for $B$. Let $P$ be the given SLP
encoding a number $N$.
Given $(P,g,b)$
we can decide in $\CH_{t+c}$ whether the value of the gate $g$ on
input $N$ given by $P$ is $b$. $t$ is the depth of $g$. If $t=0$,
then $g$ is an input gate. Thus this problem is $\bitslp$
which is in $\CH_{c}$ for some $c$. If $t>0$, then we have to decide
whether the majority of the gates that are children of $g$ are $1$.
This can be done using a $\PP$-machine with oracle to $\CH_{c+t-1}$. 
We guess a gate $f$ and check using the DTM $M$ whether $f$ is a predecessor
of $g$. If not we add an accepting and rejecting path. If yes, we 
use the oracle to check whether $f$ has value $1$. If yes we accept
and otherwise, we reject.
\end{proof}

It is easy to see that checking whether the $\ell$ least significant
bits of a number given in binary are $0$ can be done in dlogtime-uniform $\TC_0$.
Thus $\divslp$ is in $\CH$ by the lemma above. 

\section{Reduction from multivariate $\degslp$ to univariate $\degslp$}\label{sec:mdegslptodegslp}

We use $\mdegslp$ to denote the multivariate variant of the $\degslp$
problem, which we define formally below.
\begin{problem}[$\mdegslp$]
Given a straight line program representing a polynomial $f\in\IZ[\lst xn]$,
and given a natural number $d$ in binary, decide whether $\deg(f)\leq d$.
\end{problem}

$\mdegslp$ was simply called $\degslp$ in \cite{Allender06onthe}.
Now we recall the proof in \cite{Allender06onthe}, to show that to study the hardness of $\mdegslp$, it is enough
to focus on its univariate variant $\degslp$. To this end, we note
the following \ref{obs:mdegslptodegslp}.

\begin{observation}[\cite{Allender06onthe}]\label{obs:mdegslptodegslp}
$\mdegslp$ is equivalent to $\degslp$ under deterministic polynomial
time many one reductions.
\end{observation}

\begin{proof}
We only need to show that $\mdegslp$ reduces o $\degslp$ under deterministic
polynomial time many one reductions, other direction is trivial. Suppose
we are given an SLP of size $s$ which computes $f\in\IZ[\lst xn]$,
and we want to decide whether $\deg(f)\leq d$ for a given $d\in\IN$.
Suppose $D=\deg(f)$. For all $i\in\{0,1,\dots,D\}$, we use $f_{i}$
to denote the homogeneous degree $i$ part of $f$. Now notice that
for any $\boldsymbol{\alpha}=(\lst{\alpha}n)\in\IZ^{d}$, we have:

\[
f(y\boldsymbol{\alpha})=f(\lst{y\alpha}n)=\sum_{i=0}^{D}y^{i}f_{i}(\boldsymbol{\alpha}),
\]
where $y$ is a fresh variable. So if $\boldsymbol{\alpha}$ is chosen
such that $f_{D}(\boldsymbol{\alpha})$ is non-zero then $\deg(f(y\boldsymbol{\alpha}))=\deg(f)=D$.
If we choose $\alpha_{i}=2^{2^{is^{2}}}$ then it can seen that $f_{D}(\boldsymbol{\alpha})$
is non-zero, see e.g. \cite[Proof of Proposition 2.2]{Allender06onthe}.
SLPs computing $\text{\ensuremath{\alpha_{i}}}$ can be constructed
using iterated squaring in polynomial time. Hence we can construct
an SLP for $f(y\boldsymbol{\alpha})$ in polynomial time. By this
argument, we know that $\deg(f(y\boldsymbol{\alpha}))\leq d$ if and
only if $\deg(f)\leq d$ . Therefore $\mdegslp$ reduces to $\degslp$
under  polynomial time many one reductions.
\end{proof}

\end{document}